\documentclass{article}

\usepackage{arxiv}

\usepackage[utf8]{inputenc} 
\usepackage[T1]{fontenc}    
\usepackage{hyperref}       
\usepackage{url}            
\usepackage{booktabs}       
\usepackage{amsfonts}       
\usepackage{nicefrac}       
\usepackage{microtype}      
\usepackage{lipsum}
\usepackage{graphicx}
\graphicspath{ {./images/} }

\usepackage{cite}
\usepackage{amsmath,amssymb,amsfonts}
\usepackage{textcomp}
\usepackage{mathrsfs}
\usepackage{dsfont}
\usepackage[normalem]{ulem}
\usepackage{booktabs}
\usepackage{algorithm,algorithmic}
\usepackage{graphicx}
\usepackage{subcaption}
\usepackage{mwe}
\renewcommand{\epsilon}{\varepsilon}
\renewcommand{\rho}{\varrho} 

\newcommand{\Aa}{\mathcal{A}}

\newcommand{\F}{\mathcal{F}}

\newcommand{\U}{\mathcal{U}}
\newcommand{\X}{\mathcal{X}}
\newcommand{\uu}{\mathcal{U}}
\newcommand{\xx}{\mathcal{X}}
\newcommand{\ff}{\mathcal{F}}
\newcommand{\E}{\mathbb{E}}
\newcommand{\ppb}{\mathbb{P}}
\newcommand{\ppp}{\mathbf{P}}

\newcommand{\R}{\mathbb{R}}

\newcommand{\uf}{\mathfrak{u}}

\newcommand{\ub}{\mathbf{u}}

\newtheorem{prob}{Problem}
\newtheorem{definition}{Definition}

\newtheorem{remark}{Remark}
\newtheorem{prop}{Proposition}
\newtheorem{thm}{Theorem}

\newtheorem{ass}{Assumption}

\newtheorem{proof}{Proof}

\DeclareMathOperator*{\argmin}{arg\,min}
\usepackage{booktabs}

\title{Stochastic Control Barrier Functions with Bayesian Inference for Unknown Stochastic Differential Equations}

\author{
 Chuanzheng Wang \\
  Department of Applied Mathematics\\
  University of Waterloo, Canada\\
   \And
 Yiming Meng \\
  Grainger College of Engineering\\
  University of Illinois Urbana-Champaign\\
  \And
 Jun Liu \\
  Department of Applied Mathematics\\
  University of Waterloo, Canada\\
\And
  Stephen Smith\\
  Electrical and Computer Engineering\\
  University of Waterloo, Canada\\
}

\begin{document}
\maketitle
\begin{abstract}
Control barrier functions are widely used to synthesize safety-critical controls.  However, the presence of Gaussian-type noise in dynamical systems can generate unbounded  signals and potentially result in severe consequences. Although research has been conducted in the field of safety-critical control for stochastic systems, in many real-world scenarios, we do not have precise knowledge about the stochastic dynamics. In this paper, we delve into the safety-critical control for stochastic systems where both the drift and diffusion components are unknown. We employ Bayesian inference as a data-driven approach to approximate the system. To be more specific, we utilize Bayesian linear regression along with the central limit theorem to estimate the drift term, and employ Bayesian inference to approximate the diffusion term. Through simulations, we verify our findings by applying them to a nonlinear dynamical model and an adaptive cruise control model.
\end{abstract}


\label{sec:introduction}
For some real-world control problems, safety-critical control must be used in order to prevent severe consequences. It requires not only achieving control objectives, but also providing control actions with guaranteed safety \cite{garcia2015comprehensive}. Hence incorporating safety criteria is of great importance in practice when designing controllers. These requirements need to be satisfied in practice both with or without noise and disturbance. Since the notion of safety-critical control was first introduced in \cite{lamport1977proving}, there has been extensive research in safety verification problems, e.g., using discrete approximations \cite{ratschan2007safety} and computation of reachable sets \cite{girard2006efficient}. 


Recently, control barrier functions (CBFs) have been widely used to deal with safety-critical control \cite{ames2019control}. Quadratic programming (QP) problems are used to solve safety-critical control with constraints from CBFs, together with control Lyapunov functions (CLFs) for achieving stability objectives \cite{ames2014control}. The authors show that the safety criteria can be transformed into linear constraints of the QP problems. By taking the derivative of the CBF, the control inputs can be treated as the decision variables of the QP problem so that we can find a sequence of actions that can guarantee safe trajectories. The authors in \cite{rauscher2016constrained} then show that finding safe control inputs by solving QP problems can be extended to an arbitrary number of constraints and any nominal control law. Consequently, safety-critical control using QP with CBF conditions have been applied in a wide range of applications such as lane keeping \cite{ames2016control} and obstacle avoidance \cite{chen2017obstacle}. However, for many applications in robotics \cite{hsu2015control}, the derivative of the CBFs is not dependent on the control input, and thus the control action can not be directly solved using QP. To address this issue, CBFs have been extended to exponential control barrier functions (ECBF) to handle high relative degree constraints using input-output linearization \cite{nguyen2016exponential} and in \cite{xiao2019control} the authors propose a more general form of high-order control barrier functions (HOCBFs). 

However, in practice, models that are used to design controllers are imperfect and this may lead to unsafe or even dangerous behavior of the systems. Consequently, designing controllers considering uncertainty is important in practical applications. In \cite{taylor2020learning} and \cite{wang2021learning}, a bounded disturbance is considered for the model, in which the time derivative of the barrier function is separated into the time derivative of the nominal barrier function and a remainder that can be approximated using neural networks. For systems driven by Gaussian-type  noise, stochastic differential equations (SDEs) are usually used to characterize the effect of randomness. Studies for stochastic stability of diffusion-type stochastic differential equations have seen a variety of applications in verifying probabilistic quantification of safe set invariance \cite{kushner1967stochastic}. Control barrier functions for stochastic systems have also been studied in recent years. The authors in \cite{sarkar2020high} applied the strong set-invariance certificate from \cite{clark2019control} to high-order stochastic control systems using stochastic reciprocal control barrier functions (SRCBFs) and \cite{santoyo2021barrier} investigates the worst-case safety verification utilizing stochastic zeroing control barrier functions (SZCBFs) regardless of the magnitude of noise. In \cite{wang2021safety}, the authors proposed the stochastic control barrier function (SCBF) with milder conditions at the cost of sacrificing the almost sure safety. 

However, in some practical scenarios, we do not have precise information about the system and in this case, we cannot directly apply (stochastic) CBFs to control the system safely. One way of handling these unknown dynamics is to identify the model first and then applied control barrier certificates to guarantee safety.  The Matrix-Variate Gaussian process model is used to learn nonlinear control-affined discrete systems and the learned model is incorporated into a multi-agent CBF framework as in \cite{cheng2020safe}. A two-stage solution is proposed in \cite{jagtap2020control} for unknown systems, where Gaussian learning is used to learn the model first and then use CBF for safety guarantee. In \cite{dhiman2021control} and  \cite{brunke2022barrier}, Bayesian learning is used to approximate the unknown model with barrier certificates for safety-critical control. However, all of these methods only focus on deterministic systems. Data-driven methods with scenario convex programming (SCP) is also studied in guaranteeing safety of unknown dynamical systems. In \cite {chen2023data} and  \cite{salamati2021data}, random trajectories are sampled with barrier certificates as finite number of constraints in the optimization problems for unknown deterministic system and stochastic systems, respectively. However, such methods only apply to discrete systems. Formal methods with Gaussian process learning is also studied in \cite {wajid2022formal} for discrete systems as well. In \cite{wang2022data}, the authors assume that the diffusion term of the SDE is unknown and use sampling method with central limit theorem to estimate the generator the such SDEs for safety-critical control. However, such sampling-based method is very time-inefficient when the relative degree is higher than one. On the other hand, data-driven methods with Bayesian learning has been widely used to identify unknown stochastic differential equations recently as in \cite{archambeau2007gaussian} and \cite{sarkka2006recursive}. And recently, the authors in \cite{castillo2015bayesian} proposed sparse Bayesian learning architecture to handle insufficient data resources or prior information for unknown SDEs. Motivated by such results,  we consider safety-critical control problems for unknown stochastic systems. To the best of our knowledge, this is the first paper to address safety-control problems for SDEs with fully unknown drift and diffusion terms. 

In this paper, we extend our work in \cite{wang2021safety} for stochastic control barrier functions (SCBFs) to unknown SDEs. We use Bayesian learning to estimate the system and then use SCBFs on the identified system. The proposed method relies on data-driven method such that we will estimate both drift and diffusion terms based our observations of the system. We also compare the simulation results both in safety ratio and running time against our previous study in \cite{wang2021safety} and \cite{wang2022data}. As a result, the main contribution of this paper is as follows.

\begin{itemize}
    \item We propose the stochastic control barrier function with high relative degree for safety-critical control of SDEs.
    \item We propose a Bayesian learning framework for safety-critical control of unknown SDEs
    \item We validate our results and compare with our previous work in \cite{wang2021safety} and \cite{wang2022data} for safety ratio and running time. 
\end{itemize}

\section{Preliminary and Problem Definition} \label{sec:definition}
\subsection{System Description}
Given a  state space $\mathcal{X}\subseteq\R^n$, a (compact) set of control values $\U\subset \R^p$,
consider the following form of perturbed SDE:
\begin{equation}\label{E: sys}
    dX_t=f(X_t)dt+g(X_t)\uf_tdt+b(X_t)dW_t,
\end{equation}
where $W$ is an $m$-dimensional standard Wiener process;  $f:\mathcal{X}\rightarrow\R^n$ is a nonlinear local Lipschitz continuous vector field; $g, b:\mathcal{X}\rightarrow \R^{n\times m}$ are smooth mappings. The input $\uf:\R_{\geq 0}\rightarrow \U$ is a locally bounded  control signal. 


 In brief, without loss of generality, we assume that a stochastic process $X:=\{X_t\}_{t\in[0, \infty)}$ and a  process of control values $\uf:=\{\uf_t\}_{t\in[0, \infty)}$ are defined on some (unknown) probability space where the noise is generated. Given any measurable process $\uf$, the probability law of the joint process $(X,\uf):=\{(X_t,\uf_t)\}_{t\in[0, \infty)}$ can be determined on the 
           measurable space   $((\xx\times\uu)^{\otimes[0,\infty)}, \ff,  \ppp)$,  where $\ff$ is the Borel $\sigma$-algebra of $(\xx\times\uu)^{\otimes[0,\infty)}$. 
We also denote $X^\uf$ by the controlled process if we emphasize on the state-space marginal of $(X,\uf)$. 

\begin{definition}[ Control strategy]
A control strategy is a set-valued measurable function
\begin{equation}
   \kappa:\X\rightarrow 2^{\U}. 
\end{equation}
\end{definition}
We consider the following memoryless set of non-randomized control signals throughout this paper. 

\begin{definition}[State-dependent control]\label{E: control constraint}
We say that a control signal  $\uf$ conforms to a control strategy $\kappa$ for (\ref{E:sys}) if, given $\{X_t = x\}$ for any $x\in\xx$, we have that  
\begin{equation}
    \uf_t \in \kappa(x),\quad\forall t\geq 0.
\end{equation}
We further denote the set of all control signals that confirm to $\kappa$ as $\ub_\kappa$.
\end{definition}

In this view, for system \eqref{E: sys}, we only consider the weak sense of solutions. The system admits a weak solution if there exists a filtered probability space $(\Omega^\dagger,\mathscr{F}^\dagger,\{\mathscr{F}^\dagger_t\},  \ppb^\dagger)$ with a pair $(X, \uf, W)$ of adapted stochastic processes, such that $W$ is a  Wiener process and $(X, \uf)$ solves the controlled SDE \eqref{E: sys}.

Given a filtered probability space $(\Omega,\F,\{\F_t\}, \mathbb{P})$ with a natural filtration, a state space $\mathcal{X}\subseteq\R^n$, a (compact) set of control values $\U\subset \R^p$, consider a continuous-time stochastic process $X:[0,\infty)\times\Omega\rightarrow\X$ that solves the SDE
\begin{equation}\label{E:sys}
dX_t=(f(X_t)+g(X_t)u(t))dt+\sigma(X_t)dW_t,
\end{equation}
where $u:\R_{\geq 0}\rightarrow\U$ is a bounded measurable control signal; $W$ is a $d$-dimensional standard $\{\F\}_t$-Brownian motion; $f:\X\rightarrow \R^n$ is a nonlinear vector field; $g:\X\rightarrow\R^{n\times p}$ and $\sigma:\X\rightarrow\R^{n\times d}$ are smooth mappings. 

\begin{ass}\label{ass: usual}
We make the following assumptions on system \eqref{E:sys} for the rest of this paper:
\begin{itemize}
    \item[(i)] There is a $\xi\in\X$ such that $\mathbb{P}[X_0=\xi]=1$;
    \item[(ii)] The mappings $f,g,\sigma$ satisfy local Lipschitz continuity and a linear growth condition.
\end{itemize}
\end{ass}
\begin{definition}[Strong solutions]
A stochastic process $X$ is said to be a strong solution to \eqref{E:sys} if it satisfies the following integral equation
\begin{equation}
    X_t=\xi+\int_0^t(f(X_s)+g(X_s)u(s))ds+\int_0^t \sigma(X_s)dW_s, 
\end{equation}
where the stochastic integral is constructed based on the given Brownian motion $W$. 
\end{definition}

\begin{definition} [Infinitesimal generator of $X_t$]
Let $X$ be the strong solution to \eqref{E:sys}, the infinitesimal generator $\Aa$ of $X_t$ is defined by 
\begin{equation}
    \mathcal{A}h(x)=\lim\limits_{t\downarrow 0}\frac{\E^x[h(X_t)]-h(x)}{t};\;\;x\in\R^n,
\end{equation}
where $h:\R^n\rightarrow \R$ is in a set $\mathcal{D}(\Aa)$ (called the domain of the operator $\mathcal{A}$) of functions  such that the limit exists at $x$.
\end{definition}
\begin{prop}[Dynkin]
Let $X$ solve \eqref{E:sys}. If $h\in C_0^2(\mathbb{R}^n)$ then $h\in\mathcal{D}(\mathcal{A})$ and 
\begin{equation}
    \mathcal{A}h(x)=\frac{\partial h}{\partial x}(f(x)+g(x)u(t))+\frac{1}{2}\sum\limits_{i,j}\left(\sigma\sigma^T\right)_{i,j}(x)\frac{\partial^2h}{\partial x_i\partial x_j}.
\end{equation}
\end{prop}
\begin{remark}
The solution $X$ to \eqref{E:sys} is right continuous and satisfies strong Markov properties, and for any finite stopping time $\tau$ and $h\in C_0^2(\R^n)$, we have the following Dynkin's formula
$$\E^\xi[h(X_\tau)]=h(\xi)+\E^\xi\left[\int_0^\tau \Aa h(X_s)ds\right],$$
and therefore
$$h(X_\tau)=h(\xi)+\int_0^t\Aa h(X_s)ds+\int_0^\tau \nabla_x h(X_s)dW_s.$$
The above is an analogue of the evolution of $h$ along trajectories
$$h(x(t))=h(\xi)+\int_0^t [L_fh(x(s))+L_gh(x(s))u(s)]ds$$
driven by deterministic dynamics $\dot{x}=f(x)+g(x)u$, where $L_fh=\nabla_x h(x)\cdot f(x)$ and $L_gh=\nabla_x h(x)\cdot g(x)$.
\end{remark}

\subsection{Set Invariance and Control}
In deterministic settings, a set $\mathcal{C}\subseteq \X$ is said to be invariant for a dynamical system $\dot{x}=f(x)$ if, for all $x(0)\in\mathcal{C}$, the solution $x(t)$ is well defined and $x(t)\in \mathcal{C}$ for all $t\geq 0$. As for stochastic analogies, we have the following probabilistic characterization of set invariance.
\begin{definition}[Probabilistic set invariance]
Let $X$ be a stochastic process. A set $\mathcal{C}\subset\X$ is said to be invariant w.r.t. a tuple $(x, T, p)$ for $X$, where $x\in \mathcal{C}$, $T\ge 0$, and $p\in[0,1]$, if $X_0=x$ a.s. implies 
\begin{equation}
    \mathbb{P}_x[X_t\in\mathcal{C}, \;0\leq t\leq T]\geq p.
\end{equation}
Moreover, if $\mathcal{C}\subset\X$ is invariant w.r.t. $(x, T, 1)$ for all $x\in\mathcal{C}$ and $T\geq 0$, then $\mathcal{C}$ is strongly invariant for $X$.
\end{definition}

For stochastic dynamical systems with controls such as system \eqref{E:sys}, we would like to define similar probabilistic set invariance property for the controlled processes. Before that, we first define the
following concepts.
\begin{definition}[Control strategy]
A control strategy is a set-valued function
\begin{equation}
   \kappa:\X\rightarrow 2^{\U}. 
\end{equation}
\end{definition}

\begin{definition}[Controlled probabilistic invariance]
Given system \eqref{E:sys} and a set of control signals $\ub$, a set $\mathcal{C}\subset\X$ is said to be controlled invariant under $\ub$ w.r.t. a tuple $(x, T, p)$ for system \eqref{E:sys}, if for all $u\in\ub$, $\mathcal{C}$ is invariant w.r.t. $(x, T, p)$ for $X$, where $X$ is the solution to \eqref{E:sys} with $u$ as input. 

Similarly, $\mathcal{C}\subset\X$ is strongly controlled invariant under $\ub$ if $\mathcal{C}\subset\X$ is controlled invariant under $\ub$ w.r.t. $(x, T, 1)$ for all $x\in\mathcal{C}$ and $T\geq 0$. 
\end{definition}

\begin{subsection}{Problem Formulation}
For the rest of this paper, we consider a safe set of the form
\begin{equation}\label{E: safeset}
    \mathcal{C}:=\{x\in\X: h(x)\geq 0\},
\end{equation}where $h\in C^2(\R^n)$. We also define the boundary and interior of $\mathcal{C}$ explicitly as below
\begin{equation}
    \partial\mathcal{C}:=\{x\in\X: h(x)= 0\},
\end{equation}
\begin{equation}
    \mathcal{C}^\circ:=\{x\in\X: h(x)> 0\}.
\end{equation}

The objective of this paper is to control the stochastic system (\ref{E:sys}) with unknown drift and diffusion terms to stay inside the safe set. The problem is defined as follows.

\begin{prob}\label{prob:main}
Given system \eqref{E:sys} with $f$, $g$ and $\sigma$ unknown, a compact set $\mathcal{C}\subseteq\X$ defined by \eqref{E: safeset}, a point $x\in\mathcal{C}^\circ$, and a $p\in[0,1]$, design a (deterministic) control strategy $\kappa$ such that under $\ub_\kappa$, the interior $\mathcal{C}^\circ$ is controlled $p$-invariant for the resulting  solutions to \eqref{E:sys}.
\end{prob}
\end{subsection}

\begin{remark}
In our previous paper \cite{wang2022data}, we assume that we know the drift term of the system and only the diffusion term $\sigma$ is unknown. In this paper, we extend our assumption that the system is fully unknown. 
\end{remark}

\section{Safe-critical Control Design via Barrier Functions}
Before proposing our stochastic barrier certificates that can be used to design a
control strategy $\kappa$ for Problem~\ref{prob:main}, it is necessary
to review (stochastic) control barrier functions to interpret (probabilistic) 
set invariance. Note that we consider the safe set as constructed in \eqref{E: safeset}, where the function $h$ is given a priori. 
\subsection{Stochastic Reciprocal and Zeroing Barrier Functions}
Similar to the terminology for deterministic cases \cite{ ames2016control}, we introduce the construction of stochastic control barrier functions as follows.
\begin{definition}[SRCBF]
A function $B:\mathcal{C}^\circ\rightarrow \R$ is called a stochastic reciprocal control barrier function (SRCBF) for system \eqref{E:sys} if $B\in\mathcal{D}(\mathcal{A})$ and satisfies the following properties:
\begin{itemize}
    \item[(i)] there exist class-$\mathcal{K}$ functions $\alpha_1,\alpha_2$ such that for all $x\in\X$ we have
    \begin{equation}
        \frac{1}{\alpha_1(h(x))}\leq B(x)\leq \frac{1}{\alpha_2(h(x))};
    \end{equation}
    \item[(ii)] there exists a class-$\mathcal{K}$ function $\alpha_3$ such that
    \begin{equation}\label{E: rcbf strate}
        \inf\limits_{u\in\U}[\mathcal{A}B(x)-\alpha_3(h(x))]\leq 0.
    \end{equation}
\end{itemize}
We refer to the control strategy generated by \eqref{E: rcbf strate} as 
\begin{equation}
    \rho(x):=\{u\in\U: \mathcal{A}B(x)-\alpha_3(h(x))\leq 0\}
\end{equation}
and the corresponding control constraint  as $\ub_{\rho}$ (see in Definition \ref{E: control constraint}).
\end{definition}
\begin{prop}[\cite{clark2019control}]
Suppose that there exists an SRCBF for system \eqref{E:sys}. If $u(t)\in\ub_\upsilon$, then for all $t\geq 0$ and  $X_0=\xi\in\mathcal{C}^\circ$, we have $\mathbb{P}_\xi[X_t\in\mathcal{C}^\circ]=1$ for all $t\geq 0$.
\end{prop}

\begin{definition}[SZCBF]
A function $B:\mathcal{C}\rightarrow \R$ is called a stochastic zeroing control barrier function (SZCBF) for system \eqref{E:sys} if $B\in\mathcal{D}(\mathcal{A})$ and 
\begin{enumerate}
    \item[(i)] $B(x)\geq 0$ for all $x\in\mathcal{C}$;
    \item[(ii)] $B(x)<0$ for all $x\notin \mathcal{C}$;
    \item[(iii)] there exists an extended $\mathcal{K}_{\infty}$ function $\alpha$ such that
\begin{equation}\label{E: zcbf certificate}
    \sup_{u\in\U}{[\mathcal{A}B(x)}+\alpha(B(x))]\geq 0.
\end{equation}
\end{enumerate}
We refer the control strategy generated by \eqref{E: zcbf certificate} as 
\begin{equation}
    \varkappa(x):=\{u\in\U: {\mathcal{A}B(x)}+\alpha(B(x))\geq 0\}
\end{equation}
and the corresponding set of constrained control signals  as $\ub_{\varkappa}$.
\end{definition}
\begin{prop}[Worst-case probabilistic quantification]
Suppose the mapping $h$ is an SZCBF with linear function $kx$ as the class-$\mathcal{K}$ function (where $k>0$),  and the control strategy as $\varkappa(x)=\{u\in\U: \mathcal{A}h(x)+kh(x)\geq 0\}$. Let $c=\sup_{x\in\mathcal{C}}h(x)$ and  $X_0=\xi\in\mathcal{C}^\circ$, then under any $u\in\ub_{\varkappa}$ 
we have the following worst-case probability estimation:
\begin{equation}
    \mathbb{P}_{\xi}\left[X_t\in\mathcal{C}^\circ,\;0\leq t\leq T\right]\geq \left(\frac{h(\xi)}{c}\right)e^{-cT}.
\end{equation}
\end{prop}
\begin{proof}
Let $s=c-h(\xi)$ and $V(x)=c-h(x)$, then $V(x)\in [0,c]$ for all $x\in \mathcal{C}$. It is clear that $\mathcal{A}V(x)=-\mathcal{A}h(x)$. For 
 $u(t)\in\ub_\varkappa$ for all $t\in[0,T]$, we have
$$\mathcal{A}V(x)\leq -kV(x)+kc.$$
By \cite[Theorem 3.1]{kushner1967stochastic}, 
 \begin{equation}
     \begin{split}
        \mathbb{P}_{\xi}\left[\sup\limits_{t\in[0,T]}V(X_t)\geq c\right]\leq 1-\left(1-\frac{s}{c}\right)e^{-cT}. 
     \end{split}
 \end{equation}
The result follows directly after this.
\end{proof}

\subsection{Motivation of Stochastic Control Barrier Functions}
The authors in \cite{sarkar2020high} constructed high-order SRCBF and have found the sufficient conditions to guarantee pathwise set invariance with probability 1.
While the results seem strong, they come with significant costs. At the safety boundary, the control inputs need to be unbounded in order to guarantee safety. (A motivated example is illustrated in \cite{wang2021safety}). On the other hand, the synthesis of controller for a high-order system via SZCBF is with mild constraints. The trade-off is that the probability estimation of set invariance is of low quality. As a result, we propose high-order stochastic control barrier functions in this paper in order to reduce the high control efforts and improve the worst-case quantification.

\subsection{High-order Stochastic Control Barrier Functions}
To obtain non-vanishing worst-case probability estimation (compared to SZCBF), we propose a safety certificate via a stochastic Lyapunov-like control barrier function \cite{kushner1967stochastic}.
\begin{definition}[Stochastic control barrier functions] A continuously differentiable function $B:\R^n\rightarrow \R$ is said to be a stochastic control barrier function (SCBF) if $B\in\mathcal{D}(\mathcal{A})$ and the following conditions are satisfied:
\begin{enumerate}
    \item[(i)] $B(x)\geq 0$ for all $x\in\mathcal{C}$;
    \item[(ii)] $B(x)<0$ for all $x\notin \mathcal{C}$;
    \item[(iii)] $
        \sup\limits_{u\in\U}\mathcal{A}B(x)\geq 0
    $.
\end{enumerate}
We refer the control strategy generated by (iii) as 
\begin{equation}
    \upsilon(x):=\{u\in\U: \mathcal{A}B(x)\geq 0\}
\end{equation}
and the corresponding set of constrained control signals  as $\ub_{\upsilon}$.
\end{definition}

\begin{remark}
Condition (iii) of the above definition is an analogue of $\sup\limits_{u\in\U}[L_fB(x)+L_gB(x)u]\geq 0$ for the deterministic settings. The consequence is such that $\mathbb{E}^\xi[B(x)]\geq B(x)$ for all $x\in\mathcal{C}$. A relaxation of condition (iii) is given in the deterministic case such that the set invariance can still be guaranteed \cite{wang2020learning},
$$\sup\limits_{u\in\U}[L_fB(x)+L_gB(x)u]\geq -\alpha(B(x)), $$
where $\alpha$ is a class-$\mathcal{K}$ function. If $\alpha(x)=kx$ where $k>0$, under the stochastic settings, the condition formulates an SZCBF and provides a much weaker quantitative estimation of the lower bound of satisfaction probability. In comparison with SZCBF, we provide the worst-case quantification in the following proposition.
\end{remark}
\begin{prop}\label{prop: main}
Suppose the mapping $h$ is an SCBF  with the corresponding control strategy  $\upsilon(x)$. Let $c=\sup_{x\in\mathcal{C}}h(x)$ and $X_0=\xi\in\mathcal{C}^\circ$, then under the set of constrained control signals $\ub_\upsilon$, we have the following worst-case probability estimation:
$$
\mathbb{P}_{\xi}\left[X_t\in\mathcal{C}^\circ,\;0\leq t<\infty\right]\geq \frac{h(\xi)}{c}.
$$
\end{prop}
\begin{proof}
Let $V=c-h(x)$, then $V(x)\geq 0$ for all $x\in\X$ and $\mathcal{A}V\leq 0$. Let $s=c-h(\xi)$ then $s\leq c$ by definition.  
The result for every finite time interval $t\in[0,T]$ is followed by \cite[Lemma 2.1]{kushner1967stochastic},
$$\mathbb{P}_{\xi}\left[X_t\in\mathcal{C}^\circ,\;0\leq t<T\right]\geq 1-\frac{s}{c}.$$
The result follows by letting $T\rightarrow\infty$.
\end{proof}
\begin{definition}
A function $B:\X\rightarrow \R$ is called a stochastic control barrier function with relative degree $r$ for system \eqref{E:sys} if $B\in\mathcal{D}(\mathcal{A}^r)$, and  $\mathcal{A}\circ\mathcal{A}^{r-1}h(x)\neq 0$ as well as  $\mathcal{A}\circ\mathcal{A}^{j-1}h(x)=0$ for $j=1,2,\ldots,r-1$ and $x\in \mathcal{C}$.
\end{definition}

If the system \eqref{E:sys} is an $r^{\text{th}}$-order stochastic control system, to steer the process $X$ to satisfy probabilistic set invariance w.r.t. $\mathcal{C}$, we recast the mapping $h$ as an SCBF with relative degree $r$. For $h\in\mathcal{D}(\mathcal{A}^r)$, 
we define a series of functions $b_0,b_j:\X\to\mathbb{R}$ such that for each $j=1,2,\dots,r$ $b_0,b_j\in\mathcal{D}(\mathcal{A})$  and
\begin{equation}\label{eq:hocbf}
\begin{split}
    b_0(x)&=h(x),\\
    b_j(x)&=\mathcal{A}\circ \mathcal{A}^{j-1}b_0(x).\\
\end{split}
\end{equation}
 We further define the corresponding superlevel sets $\mathcal{C}_j$ for $j=1,2,\dots,r$ as
\begin{equation}\label{eq:hoset}
\begin{split}
  \mathcal{C}_j&=\{x\in\mathbb{R}^n:b_{j}(x)\geq 0\}.\\
\end{split}
\end{equation}
\begin{thm}\label{thm}
If the mapping $h$ is an SCBF with relative degree $r$, the corresponding control strategy is given as $\upsilon(x)=\{u\in\U: \mathcal{A}^rh(x)\geq 0\}$. Let $c_j=:\sup_{x\in\mathcal{C}_j}b_j(x)$ for each $j=0,1,...,r$ and $X_0=\xi\in\bigcap_{j=0}^r \mathcal{C}_j^\circ$. Then under the set of constrained control signals $\ub_\upsilon$, we have the following worst-case probability estimation:
$$\mathbb{P}_\xi[X_t\in\mathcal{C}^\circ,\;0\leq t<\infty]\geq \prod\limits_{j=0}^{r-1} \frac{b_j(\xi)}{c_j}.$$
\end{thm}
\begin{proof}
We introduce the notations 
$p_j:=\mathbb{P}_\xi[X_t\in\mathcal{C}_j^\circ,\;0\leq t<\infty]$ ,  $\hat{p}_j:=\mathbb{P}_\xi[X_t\in\mathcal{C}_j^\circ,\;0\leq t<\infty\;|\;\mathcal{A}b_j\geq 0]$.

The control signal $u(t)\in\{u\in\U: \mathcal{A}^rb(x)\geq 0\}$ for all $t\geq 0$ provides $\mathcal{A}b_{r-1}\geq 0$. By Proposition \ref{prop: main},
\begin{equation}
    \begin{split}
     p_{r-1} &=\mathbb{P}_\xi[X_t\in\mathcal{C}_{r-1}^\circ,\;0\leq t<\infty]\\
     & =\mathbb{P}_\xi[X_t\in\mathcal{C}_{r-1}^\circ,\;0\leq t<\infty\;|\;\mathcal{A}b_{r-1}\geq 0]\\
     &=\hat{p}_{r-1}.
    \end{split}
\end{equation}

For $j=0,1,...,r-2$, and $0\leq t<\infty$, we have the following recursion:
\begin{equation}\label{E: recursion}
\begin{split}
    & p_j=\mathbb{P}_\xi[X_t\in\mathcal{C}_{j}^\circ]\\ & =\mathbb{E}^{\xi}[\mathds{1}_{\{X_{t}\in C_{j}^\circ\}}\mathds{1}_{\{X_{t}\in C_{j+1}^\circ\}}]+\mathbb{E}^{\xi}[\mathds{1}_{\{X_{t}\in C_{j}^\circ\}}\mathds{1}_{\{X_{t}\notin C_{j+1}^\circ\}}]\\
    & =\mathbb{P}_\xi[X_{t}\in C_{j}^\circ\;|\mathds{1}_{\{X_{t}\in C_{j+1}^\circ\}}]\cdot\mathbb{P}_\xi[X_{t}\in C_{j+1}^\circ]\\
    & +\mathbb{P}_\xi[X_{t}\in C_{j}^\circ\;|\mathds{1}_{\{X_{t}\notin C_{j+1}^\circ\}}]\cdot\mathbb{P}_\xi[X_{t}\notin C_{j+1}^\circ],
\end{split}
\end{equation}
where we have used shorthand notations $\{X_t\in\mathcal{C}_{j}\}:=\{X_t\in\mathcal{C}_{j},\;0\leq t<\infty\}$ and $\{X_{t}\notin C_{j+1}\}:=\{X_{t}\notin C_{j+1}\;\text{for\;some}\;0\leq t<\infty\}$.
Indeed, we have $$ X_t=X_t\mathds{1}_{\{X_{t}\in C_{j+1}^\circ\}}+X_t\mathds{1}_{\{X_{t}\notin C_{j+1}^\circ\}},$$
then
$$ \mathbb{E}[X_t]=\mathbb{E}[X_t\mathds{1}_{\{X_{t}\in C_{j+1}^\circ\}}]+\mathbb{E}[X_t\mathds{1}_{\{X_{t}\notin C_{j+1}^\circ\}}],$$
and \eqref{E: recursion} follows. Note that
\begin{equation}
    \begin{split}
      &\mathbb{P}_\xi[X_{t}\in C_{j}^\circ\;|\mathds{1}_{
      \{X_{t}\in C_{j+1}^\circ\}}]\\
      &\geq \mathbb{P}_\xi[X_{t}\in C_{j}^\circ\;| \mathcal{A}b_{j}\geq 0]=\hat{p}_j,
    \end{split}
\end{equation}
and therefore 
$$\mathbb{P}_\xi[X_{t}\in C_{j}^\circ\;|\mathds{1}_{\{X_{t}\in C_{j+1}^\circ\}}]\cdot\mathbb{P}_\xi[X_{t}\in C_{j+1}^\circ]\geq \hat{p}_jp_{j+1}.$$

 Now define stopping times $\tau_j=\inf\{t: b_j(X_{t})\leq 0\}$ for $j=0,1,...,r-2$, then $b_j(X_{t\wedge \tau_j})\geq 0$ a.s.. In addition,
 $$X_{t\wedge\tau_j}=\mathds{1}_{\{\tau_j\leq \tau_{j+1}\}}X_{t\wedge\tau_j}+\mathds{1}_{\{\tau_j> \tau_{j+1}\}}X_{t\wedge\tau_j}$$
 
 Assume the worst scenario, which is for all $t\geq \tau_{j+1}$, we have $\mathcal{A}b_j\leq 0$. On $\{\mathcal{A}b_j<0\}\cap \{\tau_j> \tau_{j+1}\}$, we have $b_j(X_{\tau_{j+1}})>0$ and $\mathbb{E}^{\xi}[b_j(X_{t\wedge\tau_j})]\leq b_j(X_{\tau_{j+1}})-\int_{\tau_{j+1}}^t\varepsilon(s)ds$ for some $\varepsilon:\mathbb{R}_{\geq 0}\rightarrow \mathbb{R}_{>0}$ and $t\geq \tau_{j+1}$. Therefore, the process $b_j(X_{t\wedge\tau_j})$  is a nonnegative supermartingale and 
 $\mathbb{P}_\xi[\sup_{T\leq t<\infty}b_j(X_{t\wedge\tau_j})\geq \lambda]\leq \frac{b_j(X_{\tau_{j+1}})-\int_{\tau_{j+1}}^T\varepsilon(s)ds}{\lambda}$ by Doob's supermartingale inequality. For any $\lambda>0$, we can find a finite $T\geq \tau_{j+1}$ such that $\mathbb{P}_\xi [\sup_{T\leq t<\infty}b_j(X_{t\wedge\tau_j})\geq \lambda]=0$. Since $\lambda$ is arbitrarily selected, we must have $\mathbb{P}_\xi[\sup_{T\leq t<\infty}b_j(X_{t\wedge\tau_j})>0]=0$, which means $\tau_j$ is triggered within finite time.
 On the other hand, on $\{\mathcal{A}b_j<0\}\cap \{\tau_j\leq  \tau_{j+1}\}$, $\tau_j$ has been already triggered. Therefore, $\{X_t\notin\mathcal{C}_j^\circ\;\text{for\;some}\;0\leq t<\infty\}$ a.s. given $\{\mathcal{A}b_j<0\}$.  Hence,
 \begin{equation}
     \begin{split}
         &\mathbb{P}_\xi[X_{t}\in C_{j}^\circ\;|\mathds{1}_{\{X_{t}\notin C_{j+1}^\circ\}}]\\
         &\geq \mathbb{P}_\xi[X_{t}\in C_{j}^\circ\;|\mathds{1}_{\{X_{t}\notin C_{j+1}^\circ,\;\forall t\geq \tau_{j+1}\}}]\\
         &\geq \mathbb{P}_\xi[X_{t}\in C_{j}^\circ\;|\mathds{1}_{\{\mathcal{A}b_j<0,\;\forall t\geq \tau_{j+1}\}}]=0
     \end{split}
 \end{equation}
 Combining the above, for $j=0,1,...,r-2$, we have  
 $$p_j\geq \hat{p}_jp_{j+1},$$
 and ultimately $p_0\geq \prod\limits_{j=0}^{r-1}\hat{p}_j$. 
\end{proof}
\begin{remark}\label{rem: thm}
The above result estimates the lower bound of the safety probability given the constrained control signals $\ub_\upsilon$. Based on recursion  \eqref{E: recursion}, we can  easily obtain the same result by  dropping the last term. However, we  argued that under some extreme conditions the worst case may happen. Indeed, we have assumed that $t\geq \tau_{j+1}\implies \mathcal{A}b_j\leq 0$. This conservative assumption  
is made such that within finite time $X$ will cross the boundary of each $\mathcal{C}_j$.

Another implicit condition may cause the worst-case lower bound as well, that is when $\bigcup_{j=0}^{r-1}\{\mathcal{A}b_j=0,\;0\leq t<\infty\}$ is a $\mathbb{P}_\xi$-null set. This, however, is practically possible since the controller indirectly influences the value of $\mathcal{A}b_j$ for all $j<r$, the strong invariance of the level set $\{\mathcal{A}b_j=0\}$ is not guaranteed using QP scheme. 
\end{remark}

\section{Data-driven method for safety-critical control of unknown SDEs}
In this section, we will first address the process of identifying the system with data-driven method using Bayesian inference and then, propose our QP-based control framework for safety-critical control of the unknown SDE with learned dynamics. 
\subsection{Data collection for drift terms}
In order to collect training data to estimate $f(x)$ and $g(x)$, we sample an initial state and 
 let the system evolve under control inputs. More specifically, given some $T>0$, we divide the time interval $[0,T]$ into a sufficiently refined partition $\{t_0,t_1,...,t_N\}$ with $N\in\mathbb{N}$ such that $0=t_0<t_1<\dots<t_N=T$. Define the index set as $\mathcal{I}=\{0, 1,\dots,N\}$. At each $i\in\mathcal{I}$, we simulate the system with $K$ times from $x_i$ with two control inputs $u_{i,1}$ and $u_{i,2}$ such that $u_{i,1}\neq u_{i,2}$. For each $u_{i,j}$ with $j\in\{1,2\}$, we will get $x^k_{i,j}$ for $k\in\{1,2,...,K\}$. By applying central limit theorem,  we get 
 \begin{equation*}
     dx_{i,j}=\sum_{k=1}^K\frac{[x_{i+1,j}-x_{i}|x_{i}]}{K}.
 \end{equation*}
 Accordingly, we can calculate the target value for $f(x_{i})$ and $g(x_{i})$ as
 \begin{equation}
     \begin{split}
     y_{f(x_i)}&=\frac{dx_{i,1}\cdot u_{i,2}-dx_{i,2}\cdot u_{i,2}}{(u_{i,2}-u_{i,1})(t_{i+1}-t_{i})}\\
     y_{g(x_i)}&=\frac{dx_{i,1}-dx_{i,2}}{(u_{i,1}-u_{i,2})(t_{i+1}-t_{i})}\\
     \end{split}
 \end{equation}
 
Finally we will get training data as
\begin{equation}\label{eq:training-data}
\begin{split}
    \textbf{X}&=[x_0, \dots, x_i,\dots, x_N]^T,\\
    \textbf{Y}_f&=[y_{f(x_0)},\dots, y_{f(x_i)}, \dots, y_{f(x_N)}]^T\\
    \textbf{Y}_g&=[y_{g(x_0)},\dots, y_{g(x_i)}, \dots, y_{g(x_N)}]^T\\
\end{split}
\end{equation}

\subsection{Identification of drift terms}
We use Bayesian linear regression model  \cite{box2011bayesian} to identify the drift term of the system. Define $\Phi$ as the base function and $\theta$ as the weight vector where $\Phi\in\mathbb{R}^{N\times M}$ and $\theta\in\mathbb{R}^{M\times 1}$. Here, $M$ is the number of base function. We can write the regression equation as
\begin{equation}\label{eq:bayesian-linear}
    Y=\Phi\theta+\epsilon
\end{equation}
The noise term $\epsilon$ is assumed to follow a Gaussian distribution $\mathcal{N}(0,\Psi)$ where $\Psi$ is a diagonal matrix with the $i$th element as $\sigma^2/K$. Initially, we impose a Gaussian prior $p(\theta)$ with zero mean and covariance matrix $\Sigma_0$ on weight vector $\theta$. According to Bayes' theorem, 
the posterior distribution can be calculated analytically as 
\begin{equation*}
    P(\theta|Y)=\mathcal{N}(\bar{\theta},\bar{\Sigma})
\end{equation*}
with $\bar{\theta}$ the maximum a posterior estimation (MAP) as
\begin{equation}\label{eq:weights}
    \bar{\theta}=(\Phi^T\Phi+\frac{\sigma^2}{K}\Sigma_0^{-1})^{-1}\Phi^TY
\end{equation}
and $\bar{\Sigma}$ as the posterior covariance matrix as
\begin{equation*}
    \bar{\Sigma}=\frac{\sigma^2}{K}(\Phi^T\Phi+\frac{\sigma^2}{K}\Sigma_0^{-1})^{-1}
\end{equation*}
As a result, given training data $X$, $Y_f$ and $Y_g$ as in Equation~\ref{eq:training-data}, we can approximate the drift terms $\hat{f}(x)$ and $\hat{g}(x)$. By properly choosing the base functions $\Phi_{\hat{f}}$ and $\Phi_{\hat{g}}$, we can construct models for both $\hat{f}(x)$ and $\hat{g}(x)$ according to Equation~\ref{eq:bayesian-linear} as 
\begin{equation*}
    \begin{split}
        \hat{f}(x) &= \Phi_f\theta_f+\epsilon_f\\
        \hat{g}(x) &= \Phi_g\theta_g+\epsilon_g\\
    \end{split}
\end{equation*}
and calculate weights $\bar{\theta}_f$ and $\bar{\theta}_g$ based on Equation~\ref{eq:weights}.

\subsection{Data collection for diffusion term}
After we estimate the drift terms, we can use the estimated $\hat{f}(x)$ and $\hat{g}(x)$ to identify the diffusion term. Similar as above, we also randomly sample an initial state and let the system evolve under some control $u_i$ with $i\in\mathcal{I}$ at each time step given the refined partition $\{t_0,t_1,...,t_N\}$. Specifically, at each time step, we calculate the $\xi_i$ using estimated drift term as
\begin{equation*}
    \xi_i=x_{i+1}-x_{i}-\hat{f}(x_i)\Delta t-\hat{g}(x_i)u_i\Delta t.
\end{equation*}
Then we will get the dataset $\mathcal{D}_{\xi}$ containing $[\xi_0,\dots,\xi_{N-1}]^T$ to estimate $\sigma$. 
\subsection{Identification of diffusion term}
The likelihood function given dataset $\mathcal{D}$ is
\begin{equation*}
    L(\mathcal{D}|\sigma)=\Pi_{i=0}^{N-1}P(\xi_i|\sigma),
\end{equation*}
where
\begin{equation*}
    P(\xi_i|\sigma)=\frac{1}{\sqrt{2\pi\sigma^2}}e^{-\frac{1}{2}\frac{\xi_i^2}{\sigma^2}}
\end{equation*}
Given a prior distribution $P(\sigma)$, the posterior distribution of $\sigma$ is 
\begin{equation}\label{eq:post}
    P(\sigma|\mathcal{D})\varpropto L(\mathcal{D}|\sigma) \cdot P(\sigma).
\end{equation}

Then we use the MAP method to estimate $\sigma$ by maximizing the log-posterior distribution. According to Eq (\ref{eq:post}), we can get
\begin{equation*}
    \log(P(\sigma|\mathcal{D}))=\log{L(\sigma|\mathcal{D})+\log(P(\sigma))}
\end{equation*}
By using the inverse-gamma distribution
\begin{equation}\label{eq:gamma_prior}
    P(\sigma)=\frac{\beta^{\alpha}}{\Gamma(\alpha)}\frac{1}{\sigma^{^{\alpha+1}}}e^{-\frac{\beta}{\sigma}}
\end{equation}
as the prior distribution, we can get
\begin{equation*}
    \log{P(\sigma)|\mathcal{D}}=-(N+\alpha)\log{\sigma}-\frac{1}{2\sigma^2}\sum_{i=0}^{N-1}\xi_i^2-\frac{\beta}{\sigma}+C
\end{equation*}
with 
\begin{equation*}
    C=\alpha\log{\beta}-\log{\Gamma(\alpha)}-(N-1)\log{\sqrt{2\pi}}
\end{equation*}
is a constant \cite{box2011bayesian}. 
The best estimation $\hat{\sigma}$ is obtained by taking the derivative of the log-posterior with respect to $\sigma$ as 
\begin{equation}\label{eq:log_post}
    \frac{d\log{P(\hat{\sigma}|\mathcal{D})}}{d\hat{\sigma}}=0.
\end{equation}

\subsection{QP-based control framework for learned dynamics with SCBFs}
Based on the above identification process, we can make prediction using the learned model to approximate the dynamics of an unknown SDE. For drift terms, we can make prediction as $\hat{f}(x)=\Phi_f\bar{\theta}_f$ and $\hat{g}(x)=\Phi_g\bar{\theta}_g$ and for diffusion term, we estimate noise level using $\hat{\sigma}$ as in Equation~\ref{eq:log_post}. As a result, we can control the unknown SDE using a QP-based control framework with SCBF as
\begin{equation}\label{eq:QP}
	\begin{split}
		u(x)&=\argmin_{u\in\mathbb{R}}\frac{1}{2}||u||^2,\\
		&\text{s.t.}\quad \hat{\mathcal{A}}B(x)\geq 0,
	\end{split}
\end{equation}
where
\begin{equation*}
    \hat{\mathcal{A}}B(x) =\frac{\partial B}{\partial x}(\hat{f}(x)+\hat{g}(x)u+\frac{1}{2}\sum\limits_{i,j}\left(\hat{\sigma}\hat{\sigma}^T\right)_{i,j}(x)\frac{\partial^2B}{\partial x_i\partial x_j}.
\end{equation*}

\section{Simulation results}
\subsection{Example 1}
In the first example, we test our result using a nonlinear system given by the following stochastic differential equation as in \cite{wang2022data}:
\begin{equation*}
d\begin{bmatrix}
\dot{x}_1\\
\dot{x}_2
\end{bmatrix}=\begin{bmatrix}
-0.6x_1-x_2\\
x_1^3
\end{bmatrix}dt+\begin{bmatrix}
0\\
x_2
\end{bmatrix}udt+\begin{bmatrix}
\sigma_1&0\\
0&\sigma_2
\end{bmatrix}dW.
\end{equation*}
The safe region is defined as 
\begin{equation*}
    h(x)=-x_2^2-x_1^2+1>0.
\end{equation*}
The generator of $h$ is calculated as
\begin{equation*}
    \mathcal{A}(h)=1.2x_1^2+x_1x_2-2x_1^3x_2-2x_2^2u-(\sigma_2^2+\sigma_2^2),
\end{equation*}
and we select 
\begin{equation*}
    \begin{split}
        b_0(x)&=h(x),\\
        b_1(x)&=\mathcal{A}(x).
    \end{split}
\end{equation*}
In order to identify $f(x)$ in the drift term, we first randomly sample 10 initial points and at each initial point, we set $u=0$ and simulate the system with $\Delta t=0.01$. For any state $x_t=\xi$, we apply the control to the system at state $\xi$ for $K$ times to calculate the expectation of $f(x_t)$ using the central limit theorem as 
\begin{equation}
    f(\xi)+\epsilon\approx\sum_{j=0}^K\frac{[x_{t_K+1}-x_t|x_t=\xi]}{K\Delta t}.
\end{equation}
The we will get the training data for $f(x)$ as
\begin{equation*}
\begin{split}
    \bold{X}&=[\xi_1, \xi_2,...,.\xi_N]^T,\\
    \bold{Y}&=[f(\xi_1), f(\xi_2), ..., f(\xi_N)].
\end{split}
\end{equation*}
and estimate $f(x)$ using linear regression. 
After $f(x)$ is estimated, we use the estimated function $\hat{f}(x)$ to identify $g(x)$ in the drift term. We set control input as $u=0.1$ since the system is control affined and $g(x)$ is independent of $u$. We collect the training data similarly as in estimating $f(x)$. 
Through out this example of the simulation, we use polynomial base functions as 
\begin{equation*}
    \Phi=[1, x_1, x_2, x_1^2, x_2^2, x_1x_2, x_1^3, x_2^3]
\end{equation*}
for both $f(x)$ and $g(x)$. The result of estimation for $f(x)$ is shown in Fig \ref{fig:example1_esimate_f}. We also calculate the mean square error (MSE) for $f_1(x)$ and $f_2(x)$ over 100 randomly sample points and show the result in Table \ref{table:example1_error}.

\begin{table}[!htbp]
  \caption{MSE for estimation of $f_1(x)$ and $f_2(x)$ with $K=10,30,100$ over 100 randomly sample points.}
  \label{table:example1_error}
  \centering
  \begin{tabular}{cccc}
    \toprule
     & $f_1(x)$ & $f_2(x)$\\
    \midrule
    K=10 & $e^{-2}$ & $5e^{-2}$\\
    K=30 & $2e^{-3}$ & $8e^{-3}$\\
    K=100 & $6e^{-4}$ & $6e^{-4}$\\
    \bottomrule
  \end{tabular}
\end{table}


\begin{figure}[htbp]
    \centering
    \includegraphics[width=.8\linewidth]{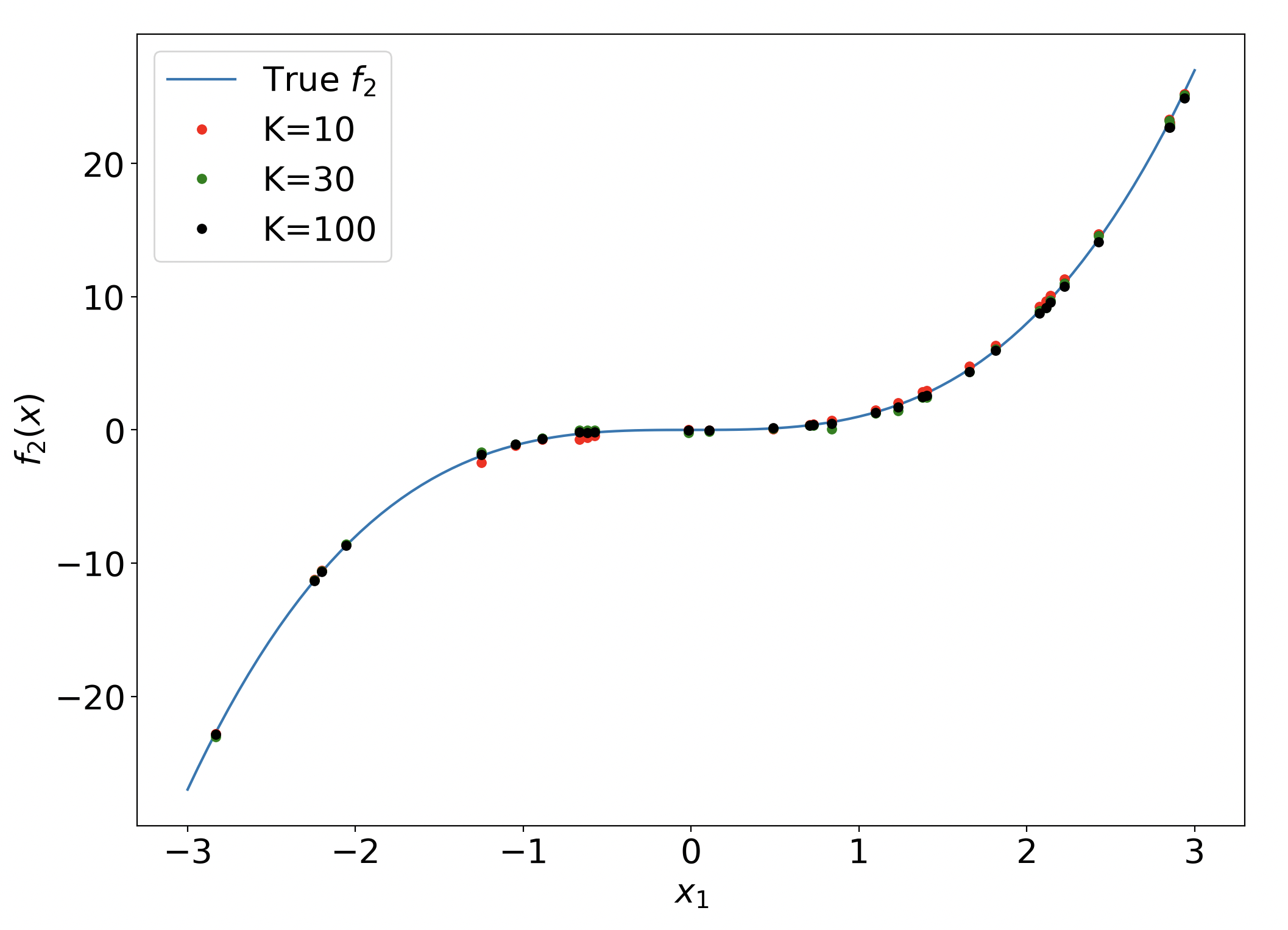}
    \caption{Estimation of $f_2(x)=x_1^3$ using $K=10$, $K=30$ and $K=100$. The estimated result is compared with the true value of the function using 30 randomly sampled points. The diffusion of the system for all the cases is $[0.2, 0.2]$.}
    \label{fig:example1_esimate_f}
\end{figure}

After we estimate the drift term, we collect data using the estimated $\hat{f}(x)$ and $\hat{g}(x)$ for 100 initial points and 300 simulation steps for each initial point. At each time step, we calculate the noise data using 
\begin{equation*}
    \phi_t=x_{t+1}-x_{t}-\hat{f}(x_t)\cdot dt-\hat{g}(x_t)\cdot u\cdot dt
\end{equation*}
and add it into the data set $\mathcal{D}$. We use inverse gamma distribution as in Equation~\ref{eq:gamma_prior} with $\alpha=1$ and $\beta=1$.

We use log-likelihood function as
\begin{equation*}
    \log L(\sigma | \mathcal{D}) = \sum_{i=1}^n \left[ -\frac{1}{2} \log(2\pi\sigma^2) - \frac{\phi_i^2}{2\sigma^2} \right],
\end{equation*}
where $n$ is the number of data in $\mathcal{D}$. We estimate a $\hat{\sigma}$ such that the log-posterior distribution 
\begin{equation*}
    \log P(\hat{\sigma} | \mathcal{D}) = \log P(\sigma) + \log L(\hat{\sigma} | \mathcal{D})
\end{equation*}
is maximized. The distribution over 10000 random sample from posterior distribution of $\sigma_1$ and $\sigma_2$ is shown in Fig \ref{fig:example1_sigma}.
\begin{figure}[htbp]
    \centering
    \includegraphics[width=.8\linewidth]{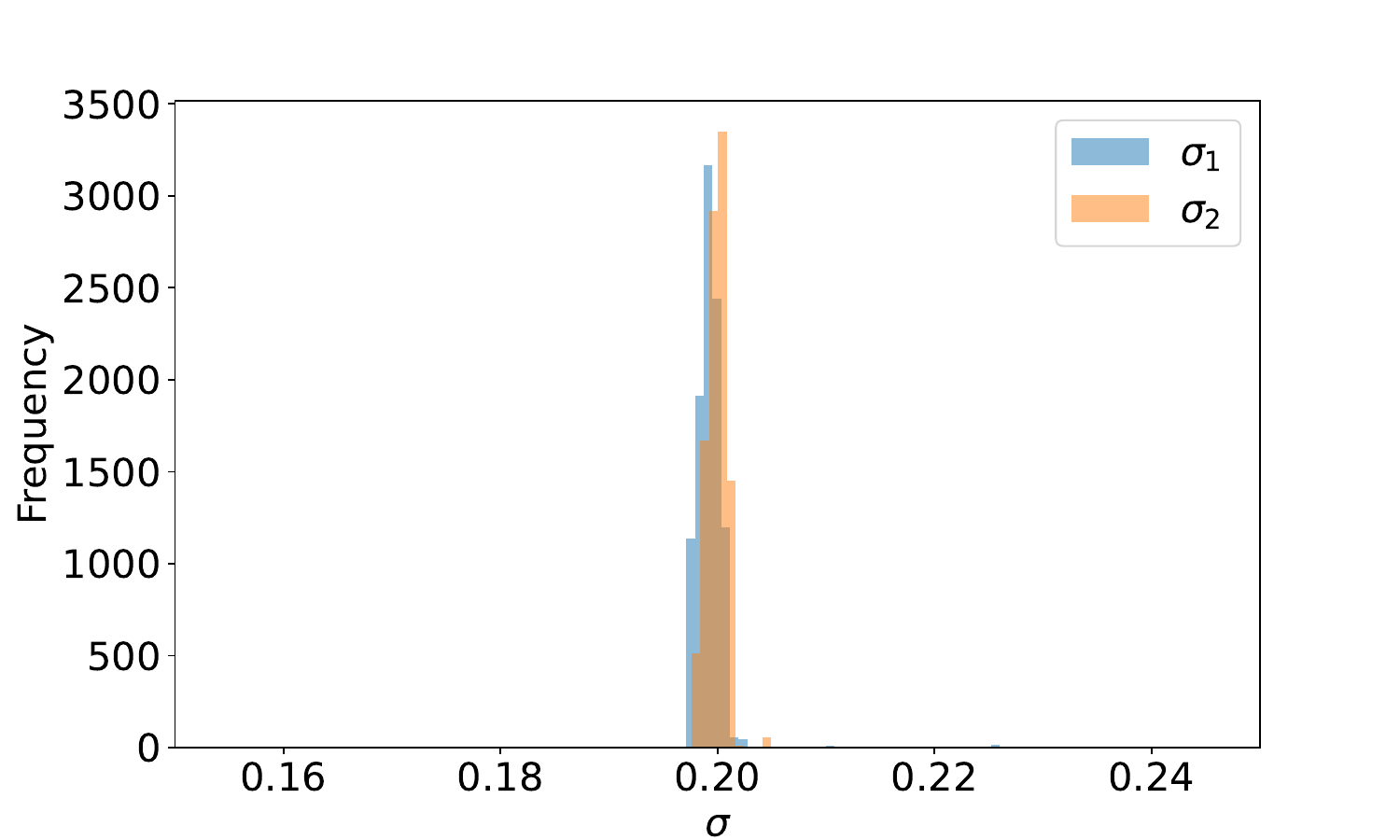}
    \caption{Posterior distribution of $\sigma_1$ and $\sigma_2$ over 10000 random samples for Example 2.}
    \label{fig:example1_sigma}
\end{figure}

In order to analyze the performance of safety control w.r.t the estimated system, we count the number of safe trajectories over 1000 simulations under different initial points between the SCBF, DDSCBF \cite{wang2022data} and Bayesian SCBF. For SCBF, we assume that we have knowledge of true system and for DDSCBF, we assume that we know the drift of the system and only the diffusion is unknown to us. The statistical results at initial state $[-0.1, 0.7]^T$ and $[-0.1, 0.8]^T$ are shown in Table \ref{table:example1_ratio}. The analytical result in the table is the worst-case safe probability calculated based on \cite[Theorem III.8]{wang2021safety}. Since the relative degree if of $r=1$, the worst-case probability is calculated as $P=\frac{b_0(x_0)}{c_0}$ where $c_0=\sup_{x\in\mathcal{C}}h(x)$ with $\mathcal{C}$ as the safe set. It is easy to find that $c_0=1$ so the worst-case probability is calculated as $P=0.5$ and $P=0.35$ for initial state $[-0.1, 0.7]^T$ and $[-0.1, 0.8]^T$, respectively. 

\begin{table}[!htbp]
  \caption{Safety ratio over 1000 simulation trajectories at initial state $[-0.1, 0.7]^T$ and $[-0.1, 0.8]^T$ with diffusion $\sigma_1, \sigma_2=0.2$.}
  \label{table:example1_ratio}
  \centering
 \begin{tabular}{cccc}
    \toprule
    & \multicolumn{2}{c}{\textbf{Safety ratio}} & \\
    \cmidrule(lr){2-3}
    & $(-0.1, 0.7)$ & $(-0.1, 0.8)$ & \\
    \midrule
    SCBF\cite{wang2021safety} & 91\% & 55\% & \\
    DDSCBF\cite{wang2022data} & 88\% & 45\% & \\
    Bayesian SCBF & 90\% & 43\% & \\
    Analytical\cite{wang2021safety} & 50\% & 35\% & \\
    
    \bottomrule
  \end{tabular}
\end{table}

\subsection{Example 2}
In the second example, we use an adaptive cruise control (ACC) model as in \cite{xiao2022sufficient} to validate the estimation of the system and safety control. The model is 
\begin{equation*}
    d\begin{bmatrix}
    v\\
    z\\
    \end{bmatrix}=\begin{bmatrix}
    -F_r(x)/M\\
    v_f-v\\
    \end{bmatrix}dt+\begin{bmatrix}
    1/M\\
    0\\
    \end{bmatrix}udt+\begin{bmatrix}
    \sigma_1& 0\\
    0 &\sigma_2\\
    \end{bmatrix} dW,
\end{equation*}
where $x=[v,z]$ is the state of the system representing the speed of the behind vehicle and the distance between two vehicles, respectively. $v_f$ is the speed of the front vehicle. The aerodynamic drag is $F_r(x)=f_0+f_1x_1+f_2x_1^2$ with $f_0=0.1$, $f_1=5$, $f_2=0.25$ and the mass of the vehicle is $M=1650$. We require the behind vehicle to reach a desired speed $v_d=22$ while keeping a minimum distance with $D=10$ from the front vehicle. Since the second-order generator of $h(x)$ is 
\begin{equation*}
    \mathcal{A}\mathcal{A}h(x)=L^2_fh(x)+L_gL_fh(x)u+\Sigma^T\frac{\partial^4 h(x)}{\partial x^4}\Sigma
\end{equation*}
with
\begin{equation*}
    \Sigma=\begin{bmatrix}
    \sigma_1& 0\\
    0 &\sigma_2\\
    \end{bmatrix},
\end{equation*}
we select $h(x)=(z-D)^5$ with a non-zero fourth-order derivative. As a result, select
\begin{equation*}
\begin{split}
    b_0(x)&=h(x),\\
    b_1(x)&=\mathcal{A}h(x)=5(z-D)^4(v_f-v),\\
    b_2(x)&=\frac{5(z-D)^4}{M}\cdot F+20(z-D)^3\cdot (v_f-v)^2\\
    &-\frac{(z-D)^4}{M}\cdot u+120(z-D)(\sigma_1^2+\sigma_2^2).\\
\end{split}
\end{equation*}
We use the same base functions, prior distribution and likelihood function as in the first example for system estimation. The result of estimation for $f(x)$ is displayed as in Fig \ref{fig:example2_esimate_f} and the distribution over 10000 random samples from the posterior distribution of $\sigma_1$ and $\sigma_2$ is shown in Fig \ref{fig:example2_sigma}.
\begin{figure}[htbp]
	\centering
	\begin{subfigure}[h]{1\linewidth}
		\includegraphics[width=.8\linewidth]{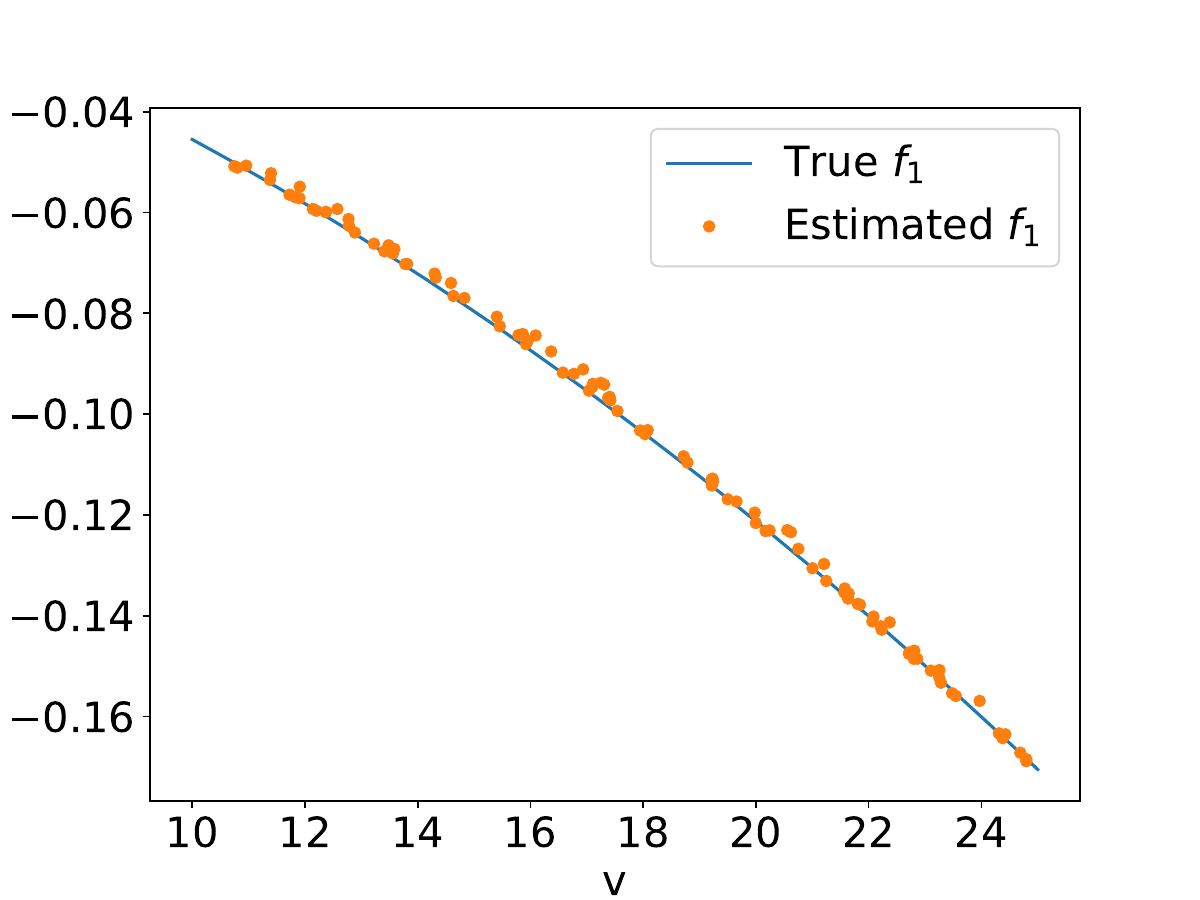}
	\caption{}
	\label{fig:nonlinear_clf}
	\end{subfigure}\\
	\begin{subfigure}[h]{1\linewidth}
		\includegraphics[width=.8\linewidth]{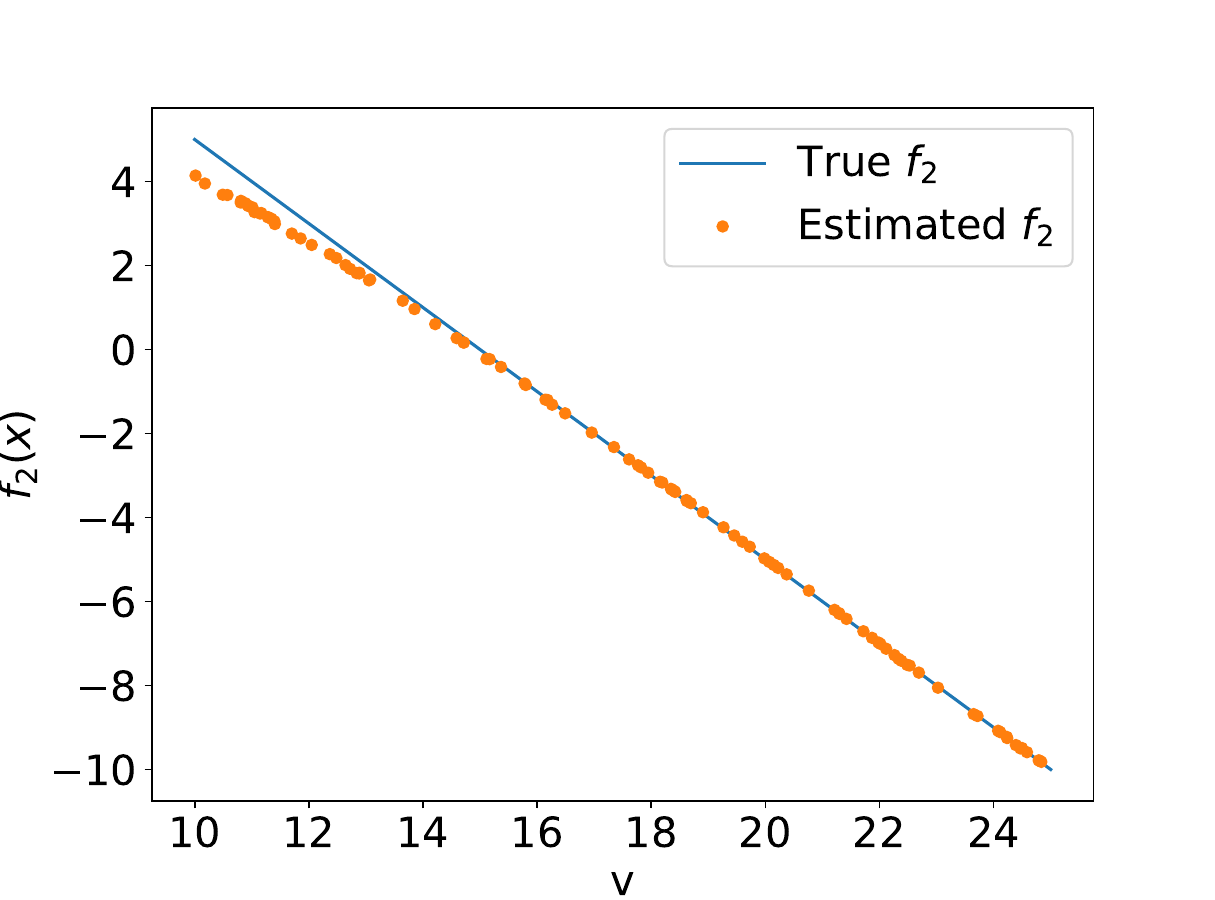}
	\caption{}
        \end{subfigure}%
	\caption{Estimation of $f(x)$ for Example 2. The result is validated using 100 randomly sampled states. (a): Estimation of $f_1(x)$. (b): Estimation of $f_2(x)$.}
    \label{fig:example2_esimate_f}
\end{figure}

\begin{figure}[htbp]
    \centering
    \includegraphics[width=.8\linewidth]{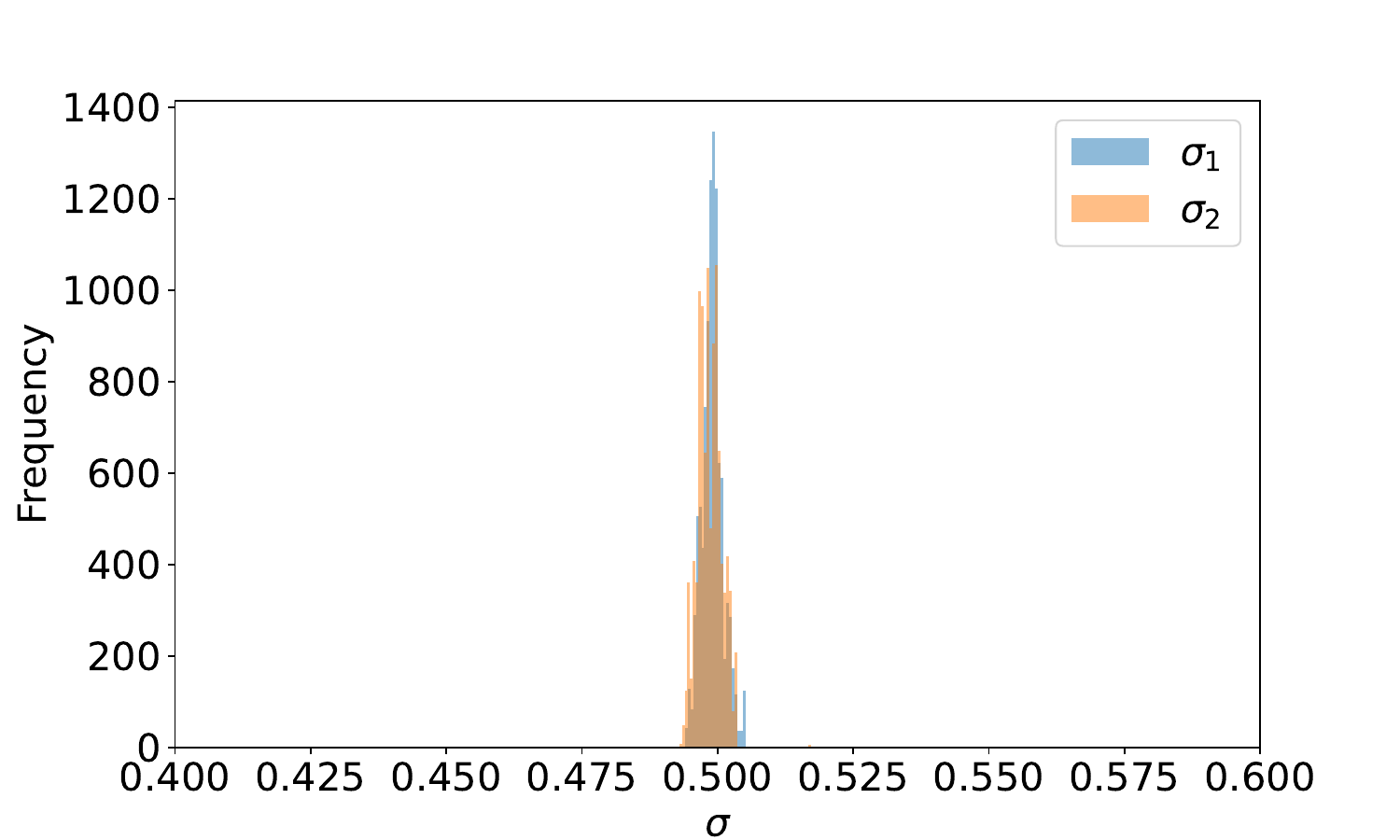}
    \caption{Posterior distribution of $\sigma_1$ and $\sigma_2$ over 10000 random samples for Example 2. The true values are $\sigma_1=0.5$ and $\sigma_2=0.5$.}
    \label{fig:example2_sigma}
\end{figure}
To control the vehicle to reach the desired speed, we use control Lyapunov function as in \cite{wang2021safety}. Similar as in example 1, we calculate the safety ratio over 1000 simulations between SCBF, DDSCBF and Bayesian SCBF. Note that the SCBF is a non-convex function in this example, we can not find the supreme value of $c_0$ as in the first example. So we only compare simulated safety-probability and the result is shown in Table \ref{table:example2_ratio}.

\begin{table}[!htbp]
  \caption{Safety ratio over 1000 simulation trajectories with diffusion $\sigma_1, \sigma_2=0.5$. The initial state is $[v, z]^T=[10, 15]^T$.}
  \label{table:example2_ratio}
  \centering
  \begin{tabular}{cccc}
    \toprule
     & Safety ratio\\
    \midrule
    SCBF\cite{wang2021safety} & 83\%  \\
    DDSCBF\cite{wang2022data} & 65\%\\
    Bayesian SCBF & 78\%\\
    \bottomrule
  \end{tabular}
\end{table}
Also, as stated in \cite{wang2022data}, the bottleneck of the DDSCBF is that when the relative degree is higher than 1, we have to recursively apply this method in each order to calculate higher order generators. As in this example, we have relative degree $r=2$, we need to learn the first-order generator, and then use the estimated first-order generator to sample data for estimating the second order generator. In the process of collecting training data for second order generator using DDSCBF, we have to sample $N$ points and at each point, sample $K$ simulations and apply CLT to calculate training data. However, in each simulation, we need to get the value of first-order generator through network, which is very time consuming. As a result, we compare the running time between DDSCBF and Bayesian estimated SCBF. The runtime of both methods are displayed in Table \ref{table:example2_runtime}. In the simulation of DDSCBF, we select $N=300$ and $K=10000$.

\begin{table}[!htbp]
  \caption{Runtime of learning process}
  \label{table:example2_runtime}
  \centering
  \begin{tabular}{cccc}
    \toprule
     & Runtime\\
    \midrule
    DDSCBF\cite{wang2022data} & over 5 hours\\
    Bayesian SCBF & 15s\\
    \bottomrule
  \end{tabular}
\end{table}

\section{conclusion}
This paper is the first one to address safety-critical control problem for fully unknown stochastic differential equations. In this paper, we use Bayesian inference to estimate both the drift and diffusion terms of the system and use SCBF to guarantee safety for the learned system. We validate safety ratio with statistical results using nonlinear models and compare with worse-case analytical safety probabilities. Future work would be focused on safety-critical control with limited observations with sparse learning inspired by recently work in SDE identifications. 

\bibliographystyle{ieeetr}
\bibliography{references}

\end{document}